\documentclass[11pt, preprint,amsmath,amssymb,letterpaper]{revtex4}

\usepackage{graphicx}
\usepackage{dcolumn}
\usepackage{bm}

\usepackage{geometry} 
\usepackage{epstopdf} 
\newcommand{\FIGS}{./figures}

\newcommand{\sHess}[3] {\partial^2 #1/\partial #2 \partial #3}

\newcommand{\bfrac}[2] {\left[ #1 / #2 \right]}

\renewcommand*{\to}{\rightarrow}

\newcommand{\set}[1]{ \left\{ #1 \right\} }

\newcommand{\rate}[1]{\stackrel{#1}{\longrightarrow}}
\newcommand{\revrate}[1]{\stackrel{#1}{\longleftarrow}}

\newtheorem{theorem}{Theorem}[section]

\newenvironment{proof}[1][Proof]{\begin{trivlist}
\item[\hskip \labelsep {\bfseries #1}]}{\end{trivlist}}

\newcommand{\qed}{\nobreak \ifvmode \relax \else
      \ifdim\lastskip<1.5em \hskip-\lastskip
      \hskip1.5em plus0em minus0.5em \fi \nobreak
      \vrule height0.75em width0.5em depth0.25em\fi}
\begin{document}

\title{Equality statements for entropy change in open systems}
\author{John M. Robinson}
\affiliation{%
Department of Biochemistry and Molecular Genetics, University of Alabama at Birmingham, Birmingham, Alabama 35294, USA 
}%
\email{jmr@uab.edu}
\date{\today}

\begin{abstract}
The entropy change of a (non-equilibrium) Markovian ensemble is calculated from (1) the ensemble phase density $p(t)$ evolved as iterative map, $p(t) = \mathbb{M}(t) p(t- \Delta t)$ under detail balanced transition matrix $\mathbb{M}(t)$, and (2) the invariant phase density $\pi(t) = \mathbb{M}(t)^{\infty} \pi(t) $. A virtual measurement protocol is employed, where variational entropy is zero, generating exact expressions for irreversible entropy change in terms of the Jeffreys measure, $\mathcal{J}(t) = \sum_{\Gamma} [p(t) - \pi(t)] \ln \bfrac{p(t)}{\pi(t)}$, and for reversible entropy change in terms of the Kullbach-Leibler measure, $\mathcal{D}_{KL}(t) = \sum_{\Gamma} \pi(0) \ln \bfrac{\pi(0)}{\pi(t)}$. Five properties of $\mathcal{J}$ are discussed, and Clausius' theorem is derived. 
\end{abstract}
\pacs{05.70.Ln, 05.20.-y, 05.40.-a}

\keywords{allosteric regulation, signal transduction, energy landscape, Markov network, Markov chain, statistical mechanics}
\maketitle

\emph{Reversible manipulation} is the principal tool of the thermodynamicist.  Reversibility appears in two forms: the quasi-static time-forward reversible transition and the microscopically-reversible  time-reversed (or adjoint) stochastic transition~\cite{Onsager:1931uq, Kurchan:1998a, Crooks:1999a}.  Application of microscopic reversibility to the path integral formulation of stochastic processes~\cite{Onsager:1953a, Lebowitz:1999a} has resulted in a set of fluctuation theorems (FT) for systems arbitrarily far from equilibrium~\cite{Harris:2007a}. FT, despite their elegance, do not provide a much needed \emph{general definition of entropy change}---an equality statement providing the entropy change for any transition---of an ensemble of Markovian systems. Here, using both the quasi-static and time-reversed transitions in the path integral approach to the dynamics of a Markovian system, we produce these equality expressions for microscopic and macroscopic entropy change.

A collection of $M$ classical particles undergoing Hamiltonian dynamics is partitioned, through scale separation, into system and bath~\cite{Zwanzig:2001a, Givon:2004kx}. The system, consisting of $N$ particles is transformed into a Markovian stochastic process described by $6N$ generalized coordinates. Phase space and time are taken as discrete quantities. Each coordinate is an $m$-tuple, and time consists of equally spaced intervals, $\Delta \tau = \tau_{i+1} - \tau_i$. The system trajectory is given as the time evolution of a phase point, ${ \sigma(\tau) = \delta \left[ (x;\tau) - (x_0;\tau) \right]}$, in phase space $\Gamma$, a $(6N\times m)$-tuple, with $x, x_0 \in \Gamma$, according to the stochastic iterative map
\begin{equation}
\sigma(\tau_i) = \mathbb{M}_{\tau_i} \sigma(\tau_{i-1}),
\label{eq:stoch_map}
\end{equation}
where $\mathbb{M}_{\tau_i} = \mathbb{M}_{\tau_i}({\sigma(\tau_i) | \sigma(\tau_{i-1}}))$ is interpreted as a stochastic matrix~\cite{Gillespie77}. Real systems, which operate under colored (OU) noise, are managed by requiring that the discrete time step in (\ref{eq:stoch_map}) be much longer than the correlation time of the noise. 
Following Gibbs, we consider an ensemble of such collections. Interpreting $\mathbb{M}_{\tau_i}$ as a transition matrix (rather than a stochastic matrix)~\cite{vanKampen:2007a} and defining the phase probability as the normalized density of phase points, $P(\tau_i) = \overline{\sigma(\tau_i)}$, the dynamics of the ensemble is a time-inhomogenious Markov chain
\begin{equation}
P(\tau_i) = \mathbb{M}_{\tau_i} P(\tau_{i-1}).
\label{eq:deterministic_map}
\end{equation}
It is assumed that $\mathbb{M}_{\tau_i}$ is a known quantity obtained through experimental parametrization or on the basis of theory. From $\mathbb{M}_{\tau_i}$ and given starting phase density $p(\tau_{0}) = P(\tau_0)$, two time-dependent quantities of interest are determined---the time-dependent phase densitiy $p(\tau_{i})$ and the time-dependent invariant phase density, $\pi(\tau_i) = \mathbb{M}_{\tau_i}^{\infty} \pi(\tau_i)$, obtained as follows: the dynamics at time $\tau_i$ are stoppped, then the density is evolved in virtual time $t_i \to \infty$ under stationary $\mathbb{M}_{\tau_i}$, according to (\ref{eq:deterministic_map}).  For nonequilibrium system ensembles, $\pi$ is a virtual quantity. For equilibrium ensembles that undergo  quasi-static perturbation, $\pi$ is a real quantity. Our results apply to ensembles that evolve according to (\ref{eq:deterministic_map}) with transition matrices $\mathbb{M}_{\tau_i}$ that are Hermitian. These systems possess three important properties~\cite{Neil:1993a}. \emph{(i)} Microscopic reversibility---$\mathbb{M}_{\tau_i}$ is self-adjoint~\cite{Onsager:1953a, Crooks:1999a}: $\mathbb{M}_{\tau_i} \pi(\tau_i) = \widetilde{\mathbb{M}}_{\tau_i} \pi(\tau_i)$. \emph{(ii)} Invariance---the invariant distribution of a stationary Markov process is independent of the ensemble history: $\lim_{t_n \to \infty} \mathbb{M}_{\tau_i}^{t_n}p(\tau_i) = \pi(\tau_i)$. \emph{(iii)} Stationarity---with invariant density: $\pi(\tau_i) = \mathbb{M}_{\tau_i} \pi(\tau_i)$.  Microscopic reversibility is a property of physical systems~\cite{Wigner:1954a} and a fundamental postulate of physics~\cite{Heisenberg:1930a}.

\begin{figure}
\centering
\includegraphics[scale=0.75]{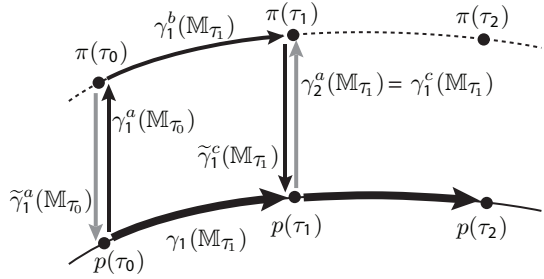}
\caption{\label{fig:prescription}
Ensemble undergoing a general perturbation. Evolution of phase density $p$ along path $\gamma$, consisting of small displacements $\gamma_i$.  The entropy change under $\gamma_i$ is measured along the virtual paths $\gamma_i^a$, $\gamma_i^b$, $\widetilde{\gamma}_i^c$ involving the invariant density $\pi$.
}
\end{figure}
\emph{General Perterbation}---With these properties in mind, we calculate the entropy change of a thermodynamic ensemble of Markovian systems undergoing an arbitrary forced perturbation. Using the prescription in Fig.~\ref{fig:prescription}, entropy change for the ensemble transition along the path increment $\gamma_i$ (bold arrows) is evaluated using three measurements that, being virtual, do not perturb the system~\cite{Heisenberg:1930a, Feynman:1948a}. The path $\gamma$ taken by one  system starting at $\sigma(t_0)$ is given by the time-ordered collection of phase points visited  by the system $\set{ \sigma(t_0)  \sigma(t_1) \cdots  \sigma(t_n)}$ under evolution by $\mathbb{M}_{t_0} \mathbb{M}_{t_1} \cdots \mathbb{M}_{t_n}$,
\begin{equation*}
\gamma \equiv \sigma(t_0) \rate{\mathbb{M}_{t_1}} \sigma(t_1) \rate{\mathbb{M}_{t_2}} \cdots  \sigma(t_{n-1}) \rate{\mathbb{M}_{t_{n}}}  \sigma(t_n).
\end{equation*}
The adjoint path is 
\begin{equation*}
\widetilde{\gamma} \equiv \widetilde{\sigma}(t_n) \revrate{\widetilde{\mathbb{M}}_{t_n}} 
\widetilde{\sigma}(t_{n-1})
\cdots \revrate{\widetilde{\mathbb{M}}_{t_2}} \widetilde{\sigma}(t_1) \revrate{\widetilde{\mathbb{M}}_{t_1}} \widetilde{\sigma}(t_0),
\end{equation*}
where $\widetilde{\mathbb{M}}_{t_j} = \widetilde{\mathbb{M}}_{t_j}({\widetilde{\sigma}(t_j) | \widetilde{\sigma}(t_{j-1}}))$. The adjoint path starts where $\gamma$ ends, $\widetilde{\sigma}(t_0) = \sigma(t_n)$. For a system in state $\sigma(t_0)$, the probability that it follows the path $\gamma$ is given by the product of the single time step transition probabilites, $\mathcal{P}(\gamma | p) = \prod_{i=i}^n \mathbb{M}_{t_{i}} $, where $p=P(\sigma(t_0))$.
Similarly, for the adjoint transition, the conditional adjoint path probability is $\mathcal{P}(\widetilde{\gamma} | \widetilde{p})  = \prod_{i=i}^n \widetilde{\mathbb{M}}_{t_{i}} $, where $\widetilde{p}=P(\widetilde{\sigma}(t_0))$ is the probability that the system starts in state $\widetilde{\sigma}(t_0)$. Using the definition of conditional probability and taking the quotient of path probabilities, we obtain
\begin{equation}
\frac{\mathcal{P}(\gamma)}{\mathcal{P}(\widetilde{\gamma})} =  \frac{ \mathcal{P}(\gamma | p) p}{\mathcal{P}(\widetilde{\gamma} | \widetilde{p}) \widetilde{p} } = \prod_{i=1}^n \frac{ \mathbb{M}_{t_i} }{ \widetilde{\mathbb{M}}_{t_i} }\frac{p}{\widetilde{p}}.
\end{equation}
The conditional path probability is also given as a function of action~\cite{Feynman:1948a, Onsager:1953a}, $\mathcal{P}(\gamma | p)  = \exp(-\sum_i \mathcal{A}_{t_i}(\gamma))$, and we note the correspondence, $\ln \mathbb{M}_{t_i} =  - A_{t_i}$.  Defining the microscopic entropy change of the collection~\cite{Lebowitz:1999a, seifert:2005a, Harris:2007a, Rahav:2007a}
\begin{equation}
\begin{array}{rcl}
\widetilde{\delta} \mathcal{S} & \equiv & \ln \mathcal{P}(\gamma)  - \ln \mathcal{P}(\widetilde{\gamma}) \\
  &  = & \ln \left[\bfrac{p}{\widetilde{p}} e^{\sum_i \ln \mathbb{M}_{t_i} -  \ln \widetilde{\mathbb{M}}_{t_i}} \right]  
\end{array},
\label{eq:var_ent}
\end{equation}
we obtain a microscopic entropy balance equation
\begin{equation}
\widetilde{\delta} \mathcal{S}  = \Delta \mathcal{S}_{\gamma}+ \ln \bfrac{p}{\widetilde{p}},
\label{eq:main}
\end{equation}
involving gain of entropy by the heat bath, $\Delta \mathcal{S}_{\gamma} = \sum_i \ln \mathbb{M}_{t_i} / \widetilde{\mathbb{M}}_{t_i}$, and gain of entropy by the system, $\ln (p / \widetilde{p})$. In (\ref{eq:var_ent}), entropy and action are on equal footing: entropy is proportional to the logarithm of exponentiated action. Evaluation of (\ref{eq:main}) is straightforward when some path $\gamma$ can be identified where $\widetilde{\delta} \mathcal{S} =0$. We show that the measurement in Fig.~\ref{fig:prescription} is along such a path.

Each path $\gamma^a_i$ involves the evolution in virtual time $t_0, t_1, \ldots, t_n$ of the ensemble starting in $p(\tau_{i-1})$ to the stationary distibution $\pi(\tau_{i-1})$. The evolution is a virtual time-homogeneous transition governed by the transition matrix $\mathbb{M}_{\tau_{i-1}}$. For each $\gamma^a_i$, (\ref{eq:main}) provides, $\Delta \mathcal{S}_{\gamma^a_i} =\widetilde{\delta} \mathcal{S}_{\gamma^a_i}  - \ln \bfrac{p(\tau_{i-1})}{\pi(\tau_{i-1})}$. Using the property of invariance \emph{(ii)}, the path probability density is 
\begin{equation}
\mathcal{P}(\gamma^a_i) = \mathbb{M}_{\tau_{i-1}}^{\infty} p(\tau_{i-1}) = \pi(\tau_{i-1}).
\label{eq:forward_path}
\end{equation}
Application of detailed balance \emph{(i)} followed by sequential application of stationarity \emph{(iii)} to the  adjoint path probability density yields
\begin{equation}
\begin{array}{ccl}
 \mathcal{P}(\widetilde{\gamma}^a_i)  &  = & \widetilde{\mathbb{M}}_{\tau_{i-1}}^{\infty}\pi(\tau_{i-1}) = \mathbb{M}_{\tau_{i-1}}^{\infty}  \pi(\tau_{i-1})\\
  &  = & \pi(\tau_{i-1}) .
\end{array}
\label{eq:adjoint_path}
\end{equation}
For each path $\gamma^a_i$, the microscopic entropy change of the collection is zero, $\widetilde{\delta} \mathcal{S}_{\gamma^a_i} = \ln \bfrac{\mathcal{P}(\gamma^a_i)  }{ \mathcal{P}(\widetilde{\gamma}^a_i) }= 0$, yielding, $\Delta \mathcal{S}_{\gamma^a_i} = - \ln \bfrac{p(\tau_{i-1})}{\pi(\tau_{i-1})}$. The macroscopic entropy over the disjoint paths $\gamma^a_i$ is the ensemble averaged entropy
\begin{equation}
\Delta \mathbf{S}_{\gamma^a} = - \sum_{i=0}^{n-1} \left< \ln  \bfrac{p(\tau_{i})}{\pi(\tau_{i})} \right>_{p(\tau_{i})},
\label{eq:macroscopic_entropy_a}
\end{equation}
where $\left< f(\Gamma) \right>_{x(\Gamma)} = \sum_{\Gamma} x(\Gamma) f(\Gamma)$.


Each path $\gamma^b_i$ involves the virtual evolution of the invariant starting distribution $\pi(\tau_{i-1})$ to the invariant distibution $\pi(\tau_{i})$ under (virtual) time-homogeneous evolution by $\mathbb{M}_{\tau_i}$. For each $\gamma^b_i$, (\ref{eq:main}) provides, $\Delta \mathcal{S}_{\gamma^b_i} =\widetilde{\delta} \mathcal{S}_{\gamma^b_i}   - \ln \bfrac{\pi(\tau_{i-1})}{\pi(\tau_{i})}$. By the same arguments used in (\ref{eq:forward_path}) and (\ref{eq:adjoint_path}), for each path $\gamma^b_i$, $\widetilde{\delta} \mathcal{S}_{\gamma^b_i} = \ln \bfrac{\mathcal{P}(\gamma^b_i)  }{ \mathcal{P}(\widetilde{\gamma}^b_i) }= 0$, yielding, $\Delta \mathcal{S}_{\gamma^b_i} =  - \ln \bfrac{\pi(\tau_{i-1})}{\pi(\tau_{i})}$. We concatenate the $\gamma^b_i$ path segments into a continuous virtual path $\gamma^b$ for the evolution of $\pi$.  The microscopic entropy over the thermodynamically reversible path $\gamma^b$ is, after cancelling terms, $\Delta \mathcal{S}_{\gamma^b} = \sum_{i=1}^n \Delta \mathcal{S}_{\gamma^b_i} =  -  \ln \bfrac{\pi(\tau_0)}{\pi(\tau_n)} $. The reversible macroscopic entropy flow\emph{ into the system} during $\gamma^b$ is
\begin{equation}
\Delta \mathbf{S}_{rev} = - \Delta \mathbf{S}_{\gamma^b} = \left< \ln \bfrac{\pi(\tau_0)}{\pi(\tau_n)}  \right>_{\pi(\tau_{0})} ,
\label{eq:rev}
\end{equation}
which is the relative, or Kullbach-Leibler, entropy $\mathcal{D}_{KL}[\mathbf{x},\mathbf{y}] = \sum_{\Gamma} \mathbf{x} \ln(\mathbf{x}/ \mathbf{y})$~\cite{Kullbach:1951a}.

Each path $\gamma^c_i$ involves the adjoint (virtual) time-homogeneous evolution of the ensemble starting from the invariant distribution  $\pi(\tau_{i})$ to the real distribution $p(\tau_{i})$ under $\widetilde{\mathbb{M}}_{\tau_i}$. For each $\gamma^c_i$ the entropy (\ref{eq:main}) is, $\Delta \mathcal{S}_{\gamma^c_i} =\widetilde{\delta} \mathcal{S}_{\gamma^c_i}   - \ln \bfrac{\pi(\tau_{i})}{p(\tau_{i})}$. Again, by the same arguments used in (\ref{eq:forward_path}) and (\ref{eq:adjoint_path}), for each path $\gamma^c_i$, $\widetilde{\delta} \mathcal{S}_{\gamma^c_i} = \ln \bfrac{\mathcal{P}(\gamma^c_i)  }{ \mathcal{P}(\widetilde{\gamma}^c_i) }= 0$, yielding, $\Delta \mathcal{S}_{\gamma^c_i} = - \ln \bfrac{\pi(\tau_i)}{p(\tau_i)} $. The macroscopic entropy change over the disjoint paths $\gamma^c_i$ is 
\begin{equation}
\Delta \mathbf{S}_{\gamma^c} = - \sum_{i=1}^{n} \left< \ln  \bfrac{\pi(\tau_{i})} {p(\tau_{i})} \right>_{\pi(\tau_{i})}.
\label{eq:macroscopic_entropy_c}
\end{equation}
From (\ref{eq:macroscopic_entropy_a}) and (\ref{eq:macroscopic_entropy_c}), the irreversible macroscopic entropy flow \emph{into the system} over $\gamma^a$ and $\gamma^c$,  $\Delta \mathbf{S}_{irrev} = - (\Delta \mathbf{S}_{\gamma^a} + \Delta \mathbf{S}_{\gamma^c})$, is
\begin{equation}
\Delta \mathbf{S}_{irrev} = B_{0} + \sum_{i=1}^{n-1} \left< \ln  \bfrac{p(\tau_{i})}{\pi(\tau_{i})}  \right>_{p(\tau_{i}) - \pi(\tau_{i})} + B_{n},
\label{eq:irrev}
\end{equation}
where  $B_{0} =\sum_{\Gamma}  p(\tau_{0}) \ln p(\tau_0) / \pi(\tau_0)$ and $B_{n} =\sum_{\Gamma}  \pi(\tau_{n}) \ln \pi(\tau_n) / p(\tau_n)$ are boundary Kullbach-Leibler integrals. The sum in (\ref{eq:irrev}) is over the Jeffreys invariant divergence measure, $\mathcal{J}[\mathbf{x},\mathbf{y}] = \sum_{\Gamma} (\mathbf{x} - \mathbf{y}) \ln(\mathbf{x}/ \mathbf{y})$~\cite{Jeffreys:1946a}. Jeffreys~\cite{Jeffreys:1961} and others \cite{Kullbach:1951a} have commented on the many remarkable properties of $\mathcal{J}$.

\begin{figure}
\centering
$\begin{array}{l}
\multicolumn{1}{l}{\mbox{\bf (a)}} \\ [-0.53cm] 
\includegraphics[scale=0.8]{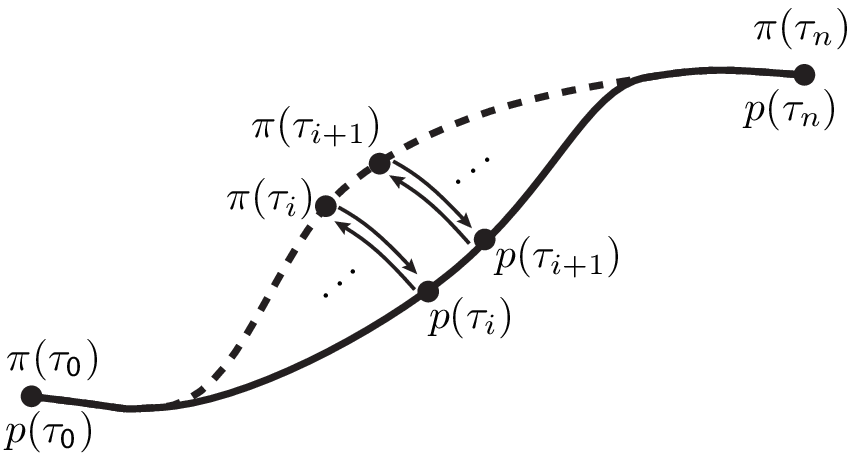}\\
\\
\multicolumn{1}{l}{\mbox{\bf (b)}} \\ [-0.53cm] 
\includegraphics[scale=0.75]{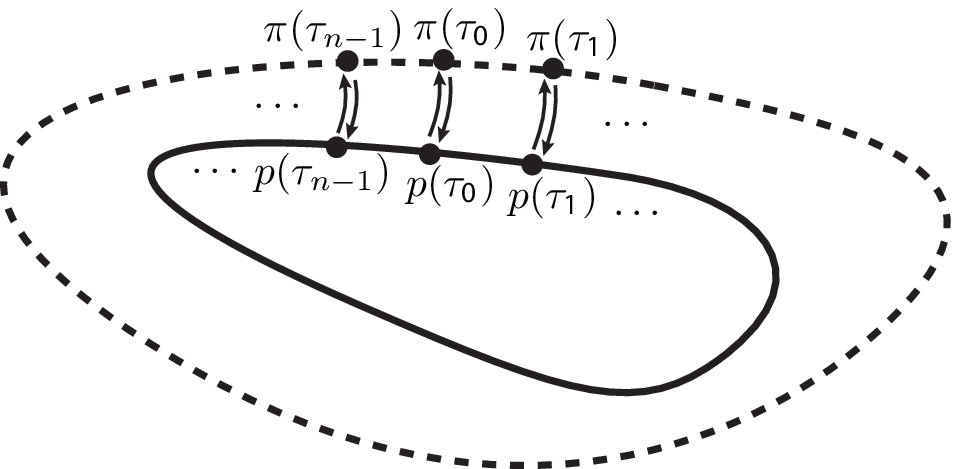}
\end{array}$
\caption{\label{fig:transformations}
Two perterbations. (a) Non-equilibriuum perturbation with $p$ starting and ending in equilibrium (convergence of $p$ and $\pi$). (b) Periodic perturbation with closed orbits of $p$ and $\pi$.
}
\end{figure}

\emph{Perturbation \#1}---We apply the results obtained for the general perturbation to two specific, and important, perturbations (Fig.~\ref{fig:transformations}). In Fig.~\ref{fig:transformations}a, the system is perturbed from one equilibrium state to another.  For the (virtual) equilibrium path (dashed line), macroscopic entropy flow into the system is, from (\ref{eq:rev}), the Kullbach-Leibler entropy, $\Delta \mathbf{S}_{rev} = \mathcal{D}_{KL}[\pi(\tau_0), \pi(\tau_n)]$. For the (real) non-equilibrium path (solid line), macroscopic entropy is the sum of the reversible entropy (\ref{eq:rev}) and the irreversible entropy (\ref{eq:irrev}), $\Delta \mathbf{S}_{tot} = \Delta \mathbf{S}_{rev} + \Delta \mathbf{S}_{irrev}$. Here, $p(\tau_0) = \pi(\tau_0)$ and $p(\tau_n) = \pi(\tau_n)$, causing $B_0$ and $B_n$ to vanish. The irreversible entropy flow into the system is the discrete time integral over the Jeffreys invariant measure, $\Delta \mathbf{S}_{irrev} = \sum_{i=1}^{n-1} \mathcal{J}[p(\tau_i),\pi(\tau_i)] $.  The total entropy flow into the system is
\begin{equation}
\Delta \mathbf{S}_{tot} = \mathcal{D}_{KL}[\pi(\tau_0), \pi(\tau_n)] + \sum_{i=1}^{n-1} \mathcal{J}[p(\tau_i),\pi(\tau_i)] .
\end{equation}

\emph{Perturbation \#2}---In Fig.~\ref{fig:transformations}b,  we consider an ensemble of systems undergoing periodic perturbation with period $\omega \Delta \tau$, where $\mathbb{M}_{\tau_i} = \mathbb{M}_{\tau_{i+\omega}}$. Using (\ref{eq:deterministic_map}), an initial phase density $p(\tau_0) = P(\sigma(\tau_{0}))$, is prepeared from some arbitrary phase density $p(\tau_{-\Omega})$ through the equilibration process: $p(\tau_0) = \left( \mathbb{M}_{\tau_{\omega}} \mathbb{M}_{\tau_{\omega -1}}  \ldots \mathbb{M}_{\tau_1} \right)^{\Omega/\omega} p(\tau_{-\Omega})$, where $\Omega/\omega \in \mathbb{I}^{+} \gg 1$. While the Poincar\'e recurrence time for any one system may be extremely long, the recurrence time for the ensemble is $\omega \Delta \tau$. From (\ref{eq:rev}) and using the property,  $\pi(\tau_0) = \pi(\tau_n)$, we obtain $\Delta \mathbf{S}_{rev} = \mathcal{D}_{KL}[\pi(\tau_0), \pi(\tau_n)] = 0$.  From (\ref{eq:irrev}) and the property, $p(\tau_0) =p(\tau_n)$, we obtain $\Delta \mathbf{S}_{irrev} = \sum_{i=1}^{n} \mathcal{J}[p(\tau_i),\pi(\tau_i)] $.  The total entropy flow into the system over one cycle of perturbation is
\begin{equation}
\Delta \mathbf{S}_{tot} = \sum_{i=1}^{n} \mathcal{J}[p(\tau_i),\pi(\tau_i)].
\label{eq:macroscopic_entropy}
\end{equation}
\emph{The total entropy transferred from the bath to the system over a thermodynamic cycle is the time integral of the Jeffreys divergence between real and invariant phase densities.}





\emph{Properties of $\mathcal{J}$}---The properties of $\mathcal{J}$ generate some important conclusions.  $\mathcal{J}$ is almost positive definite, meaning \emph{(i)} $\mathcal{J}[\mathbf{x},\mathbf{y}] \geq 0$ and \emph{(ii)} $\mathcal{J}[\mathbf{x},\mathbf{y}] = 0$ only when $x = y$. \emph{(iii)} $\mathcal{J}$ is symmetric: $\mathcal{J}[\mathbf{x},\mathbf{y}] = \mathcal{J}[\mathbf{y},\mathbf{x}]$. \emph{(iv)} $\mathcal{J}$ is a linear measure: $\det \left( \sHess{\mathcal{J}}{\mathbf{x}_i}{ \mathbf{x}_j} \right) = 0$. $\mathcal{D}_{KL}$ satisfies \emph{(i)}, \emph{(ii)}~\cite{Gray:2007a} and \emph{(iv)}. See EPAPS Document No. [] for proofs and further discussion. 

For an isothermal (cannonical) system, $\Delta \mathbf{S} = - \Delta \mathbf{S}_{tot}$, is the heat flow per temperature from the system to the bath during a periodic perturbation, $\Delta \mathbf{S} = \oint \beta d\mathbf{Q}$, where $\beta = 1/k_B T$ is inverse temperature in units of energy. Clausius' statement of the second law of thermodynamics is obtained from (\ref{eq:macroscopic_entropy}) and properties \emph{(i)} and \emph{(ii)},
\[
\oint \beta d\mathbf{Q} \leq 0.
\]

Practical application of (\ref{eq:macroscopic_entropy}) to many-body systems derives from a fifth property of $\mathcal{J}$, \emph{(v)} decomposability~\cite{Jeffreys:1961}: for a system with a decomposable Markov transition matrix
\[
\mathbb{M} = \left[ \begin{array}{cc}
\mathbb{A}  & \mathbf{0}  \\
\mathbf{0}  &   \mathbb{B} 
\end{array}
\right],
\]
the phase density decomposes: $p=p_A p_B$, $\pi=\pi_A \pi_B$, phase space decomposes: $\Gamma = \Gamma_A \Gamma_B$, and the invariant measure decomposes: $\mathcal{J}[p , \pi]  = \mathcal{J}^A[p_A , \pi_A] + \mathcal{J}^B[p_B , \pi_B]$. Defining, $\Delta \mathbf{S}^k =  \sum_i\mathcal{J}^k_i[p_i , \pi_i]$, and using \emph{(v)}, we find that macroscopic entropy is extensive, $\Delta \mathbf{S}  = \Delta \mathbf{S}^A  + \Delta \mathbf{S}^B$. The most immediate application of decomposability is the overdamped system where $\Gamma_B$ comprises momentum space and momentum is always equilibrated, $p_B = \pi_B$. Using (\ref{eq:macroscopic_entropy}) and properties \emph{(ii)} and \emph{(v)}, we obtain, $\Delta \mathbf{S}  = \Delta \mathbf{S}_A$.  Further reduction may be possible with a suitable choice of basis for $\Gamma_A$~\cite{Kitao:1999lr} and also upon coarse graining of the system~\cite{Givon:2004kx, Gohlke:2006fk}.

This work was supported by the NIH. The encouragement of Herbert Cheung; and discussions with Horacio Wio and Nikolai Chernov are gratefully acknowledged.

\clearpage
\newpage

\section*{EPAPS Document No. []}
\subsection{Some comments}

Aside from generating equality statements for reversible and irreversible entropy change, perhaps our most important contribution is  providing a purely virtual integration protocol for evaluating a variation.  Being virtual, the integration can not introduce uncertainty into the system~\cite{Feynman:1948a, Heisenberg:1930a}.  This measurement protocol is likely to find application beyond statistical mechanics.

We note that Gibbs entropy, $\mathcal{S} = - \ln p$, ~\cite{Crooks:1999a, seifert:2005a} and Boltzmann entropy, $\mathcal{S} = - \ln \pi$, follow as boundary terms from the definition of variational entropy $\widetilde{\delta} \mathcal{S}$ when $p$ or $\widetilde{p}$ are evaluated at $\pi$.

\subsection{Some properties of the Jeffreys divergence $\mathcal{J}$ and the Kullback-Leibler divergence $\mathcal{D}$}
The Jeffreys divergence measure is defined
\[
\mathcal{J}[\mathbf{x},\mathbf{y}] = \sum_{\Gamma} (\mathbf{x} - \mathbf{y}) \ln(\mathbf{x}/ \mathbf{y}).
\]
The Kullback-Leibler divergence measure is defined
\[
\mathcal{D}[\mathbf{x},\mathbf{y}] = \sum_{\Gamma} \mathbf{x}  \ln(\mathbf{x}/ \mathbf{y}).
\]

\begin{theorem}
$\mathcal{J}$ is almost positive definite, meaning \emph{(i)} $\mathcal{J}[\mathbf{x},\mathbf{y}] \geq 0$ and \emph{(ii)} $\mathcal{J}[\mathbf{x},\mathbf{y}] = 0$ only when $\mathbf{x} = \mathbf{y}$. 
\end{theorem}
\begin{proof}
We consider, element-wise, the probability densities $\mathbf{x} = \set{x}$, $\mathbf{y} = \set{y}$, and  the Jeffreys measure, $\mathcal{J}[\mathbf{x},\mathbf{y}] =  \sum_{\Gamma} \mathcal{J}_i[x,y]$. For $x > y$, $(x -y )> 0$ and $\ln(x / y) > 0$; therefore, $\mathcal{J}_i[x,y] = (x-y) \ln(x/y) >0$.  For $x < y$, $(x -y) < 0$ and $\ln(x / y) < 0$; therefore, $\mathcal{J}_i[x,y] >0$.  For $x = y$, $(x -y) = 0$ and $\ln(x / y) = 0$; therefore, $\mathcal{J}_i[x,y] = 0$. \qed
\end{proof}
\begin{theorem}
$\mathcal{J}$ is symmetric: $\mathcal{J}[\mathbf{x},\mathbf{y}] = \mathcal{J}[\mathbf{y},\mathbf{x}]$
\end{theorem}
\begin{proof}
$J$ is evaluated element-wise. Both terms $(x-y)$ and $\ln(x/y)$ are odd under exchange of $x$ and $y$. The product of two odd functions is even.  \qed
\end{proof}
\begin{theorem}
$\mathcal{J}$ does not satisfy the triangle inequality: $\mathcal{J}[x,z] \leq \mathcal{J}[x,y] + \mathcal{J}[y,z]$.
\end{theorem}
\begin{proof}
The proof is by example (Nikolai Chernov, personal communication). Let $x = [0.25, 0.75]$, $y = [0.50, 0.50]$, and $z = [0.75, 0.25]$.  $\mathcal{J}[x,y] = 0.27$; $\mathcal{J}[y,z] = 0.27$; $\mathcal{J}[x,z] = 1.10$. We obtain, $\mathcal{J}[x,z] > \mathcal{J}[x,y] + \mathcal{J}[y,z]$. \qed
\end{proof}
By not satisfying the triangle equality, the Jeffreys measure falls short of being a topologic metric~\cite{Hocking:1988a}.  For this reason, the term ``Jeffreys divergence measure" is the preferred over the ``Jeffreys distance measure."
\begin{theorem}
$\mathcal{J}$ is a linear measure.
\end{theorem}
\begin{proof}
The Hession of $\mathcal{J}[x,y]$, 
\[
H(\mathcal{J}[x,y]) = \left[ \begin{array}{cc}
\sHess{\mathcal{J}}{x}{x} & \sHess{\mathcal{J}}{x}{y} \\
\sHess{\mathcal{J}}{y}{x} & \sHess{\mathcal{J}}{y}{y} 
\end{array}
\right],
\]
is evaluated: $\sHess{\mathcal{J}}{x}{x} = (x+y)/x^2$, $\sHess{\mathcal{J}}{y}{y} = (x+y)/y^2$, and $\sHess{\mathcal{J}}{x}{y} = \sHess{\mathcal{J}}{y}{x}= -(x+y)/xy$.  By substitution we find, $\det H(\mathcal{J}) = 0$. \qed
\end{proof}

\begin{theorem}
$\mathcal{D}$ is a linear measure.
\end{theorem}
\begin{proof}
The Hession of $\mathcal{D}[x,y]$ is evaluated: $\sHess{\mathcal{D}}{x}{x} = 1/x$, $\sHess{\mathcal{D}}{y}{y} = x/y^2$, and $\sHess{\mathcal{D}}{x}{y} = \sHess{\mathcal{D}}{y}{x}= -1/y$.  By substitution we find, $\det H(\mathcal{D}) = 0$. \qed
\end{proof}

\end{document}